\documentclass[11pt,a4paper]{scrartcl}
\usepackage{amsmath}
\usepackage{amssymb}
\usepackage[donotfixamsmathbugs]{mathtools}
\usepackage{tikz}
\usepackage{xcolor,colortbl}

\usepackage{xspace}
\usepackage{paralist}

\usepackage{graphicx}

\usepackage{multirow}
\usepackage{hyperref}

\usepackage{psfrag}
\usepackage{subfigure}
\usepackage{ifthen}

\usepackage{ifpdf}
\ifpdf
\usepackage{microtype}
\fi

\newcommand{\nopappendix}[1]{}
\newcommand{\nopappendixbib}[1]{}

\newcommand{\nn}{\mathbb{N}}
\newcommand{\bigO}[1]{\ensuremath{{\mathcal O}(#1)}}

\usepackage{amsthm}
\theoremstyle{plain}
\newtheorem{theorem}{Theorem}%
\newtheorem{lemma}{Lemma}
\newtheorem{proposition}{Proposition}
\newtheorem{corollary}{Corollary}
\newtheorem{prob}{Problem}
\newtheorem{observation}{Observation}
\theoremstyle{definition}
\newtheorem{definition}{Definition}
\theoremstyle{remark}

\newtheorem*{claim}{Claim}

\newenvironment{cproof}{\vspace{-0.5cm}\paragraph{\normalfont\textit{Proof of Claim:}}}{\hfill$\blacksquare$}

\newcommand{\ceil}[1]{{\lceil #1 \rceil}}

\newcommand{\cc}[1]{{\mbox{\textnormal{\textsf{#1}}}}\xspace}  %

\newcommand{\NP}{\cc{NP}}

\newcommand{\FPT}{\cc{FPT}}

\newcommand{\paraNP}{\cc{paraNP}}

\renewcommand{\phi}{\varphi}

\newcommand{\yes}{\textsc{Yes}}

\newcommand{\PBMP}{\textup{P-BMP}\ensuremath{^e}}

\newcommand{\BMP}{\textup{BMP}\ensuremath{^e}}

\newcommand{\probeset}{{\mathcal S}}
\newcommand{\neighbor}{{\mathcal N}}

\newcommand{\embed}{{\varepsilon}}
\newcommand{\placement}{{\phi}}
\newcommand{\borderlen}{\textup{BL}}
\newcommand{\border}{\textup{border}}
\newcommand{\bd}{\textup{bd}}

\newcommand{\mask}{{\mathcal M}}
\newcommand{\verticalcost}{\textup{cost}_\textit{v}}
\newcommand{\horizontalcost}{\textup{cost}_\textit{h}}

\newcommand{\cT}{{\mathcal T}}
\newcommand{\cI}{{\mathcal I}}

\newcommand{\comment}[1]{}

\title{%
Parameterized Complexity of Asynchronous Border Minimization\thanks{Supported by the Austrian Science Fund (FWF): P25518-N23 and P26696, and the German Research Foundation (DFG) under grant ER 738/2-1.}
} %

\author{Robert Ganian$^1$, Martin Kronegger$^1$, Andreas Pfandler$^{1,2}$,\\ and Alexandru Popa$^3$ \vspace{0.4em} \\\normalsize$^1$Vienna University of Technology, Vienna, Austria\\[-4pt] \small \texttt{firstname.lastname@tuwien.ac.at} \vspace{0.4em}\\\normalsize$^2$University of Siegen, Siegen, Germany  \vspace{0.4em} \\\normalsize$^3$Nazarbayev University, Astana, Kazakhstan\\[-4pt] \small \texttt{alexandru.popa@nu.edu.kz}}
\date{March, 2015}

\begin{document}
\maketitle

\begin{abstract}
\noindent{\normalfont\sectfont Abstract.}\ \ Microarrays are research tools used in gene discovery as well as disease and cancer diagnostics. Two prominent but challenging problems related to microarrays are the \emph{Border Minimization Problem} (BMP) and the \emph{Border Minimization Problem with given placement} (\mbox{P-BMP}).

In this paper we investigate the parameterized complexity of natural variants of BMP and P-BMP, termed $\BMP$ and $\PBMP$ respectively, under several natural parameters.  We show that $\BMP$ and $\PBMP$ are in FPT under the following two combinations of parameters: $1$) the size of the alphabet ($c$), the maximum length of a sequence (string) in the input ($\ell$) and the number of rows of the microarray ($r$); and, $2$) the size of the alphabet and the size of the border length ($o$). Furthermore, $\PBMP$ is in FPT when parameterized by $c$ and $\ell$. We complement our tractability results with corresponding hardness results.

\end{abstract}

\section{Introduction}

DNA and peptide microarrays~\cite{CC09,GR+99} are important research tools
used in gene discovery, multi-virus discovery
as well as disease and cancer diagnosis.
Apart from measuring the amount of gene expression~\cite{ST+00},
microarrays are an efficient tool for making a qualitative
statement about the presence or absence of biological target sequences
in a sample. For example, peptide microarrays are used for detecting tumor biomarkers~\cite{CMI+06,MES+04,WSK+03}.

A microarray is a plastic or glass slide consisting of thousands of
sequences of nucleotides called \emph{probes}  that are assigned to one cell in the array.
The synthesis process~\cite{FR+91} consists of two components:
\emph{probe placement} and \emph{probe embedding}.
In the probe placement, the goal is to determine an assignment of each probe to a unique cell of the array.
If the placement is given one has to create the sequences at their respective cells (probe embedding).
This can be achieved with help of the following two operations: It is possible to \emph{mask} a certain set of cells. Furthermore, one can \emph{append} a certain nucleotide to the probes in all those cells which are currently unmasked.
Essentially, the nucleotides are represented as characters and the probes as strings. In probe embedding we want to find a common supersequence of all probes,
called the \emph{deposition sequence},
and a sequence of 2D arrays describing the masks.
The cells of a mask can be either masked (opaque) or
unmasked (transparent) allowing the deposition of the nucleotide associated with the mask.
For any cell, the concatenation of the nucleotides for which the cell is transparent
has to match the probe in that cell of the microarray.
See Figure~\ref{fig:masksub1} for an example~\cite{LWY+08}.

Due to diffraction, %
the cells on the \emph{border} between the masked and the unmasked regions
are often subject to unintended illumination~\cite{FR+91},
and can compromise experimental results.
Therefore, 
unintended illumination should be minimized.
The magnitude of unintended illumination can be measured by the
\emph{border length} of the masks used,
which is the number of borders shared between masked and
unmasked regions,
e.g., in Figure~\ref{fig:masksub1}, the border
length of $\mask_1,\mask_3,\mask_4$ is $2$ and $\mask_2$ is~$4$ which yields a total border length of $10$.

\begin{figure}[t]
\begin{center}
 {
  \psfrag{S}{\footnotesize{$S = ACGTA$}}
  \psfrag{e1}{\footnotesize{$e_1 = $}}
  \psfrag{e2}{\footnotesize{$e_2 = $}}
  \psfrag{e3}{\footnotesize{$e_3 = $}}
  \psfrag{e4}{\footnotesize{$e_4 = $}}
  \psfrag{e5}{\footnotesize{$e_5 = $}}
  \psfrag{M1}{\scriptsize{$\mask_1$}}
  \psfrag{M2}{\scriptsize{$\mask_2$}}
  \psfrag{M3}{\scriptsize{$\mask_3$}}
  \psfrag{M4}{\scriptsize{$\mask_4$}}
  \psfrag{M5}{\scriptsize{$\mask_5$}}
  \psfrag{A}{\tiny{A}}
  \psfrag{C}{\tiny{C}}
  \psfrag{G}{\tiny{G}}
  \psfrag{T}{\tiny{T}}
  \psfrag{AC}{\tiny{AC}}
  \psfrag{CA}{\tiny{CA}}
  \psfrag{CT}{\tiny{CT}}
  \psfrag{TA}{\tiny{TA}}
  \psfrag{unmasked region}{\tiny{unmasked region}}
  \psfrag{masked region}{\tiny{masked region}}

  \psfrag{p}{\footnotesize{$p =$ CT}}
  \psfrag{1}{\footnotesize{(a)}}
  \psfrag{2}{\footnotesize{(b)}}
  \psfrag{3}{\footnotesize{(c)}}
  \psfrag{4}{\footnotesize{(d)}}
  \psfrag{1}{}
  \psfrag{2}{}
  \psfrag{3}{}
  \psfrag{4}{}
  \psfrag{a}{\scriptsize{A}}
  \psfrag{c}{\scriptsize{C}}
  \psfrag{g}{\scriptsize{G}}
  \psfrag{t}{\scriptsize{T}}
  \psfrag{S}{\scriptsize{$D$}}
  \psfrag{E1}{{$\embed_1$}}
  \psfrag{E2}{{$\embed_2$}}
  \psfrag{E3}{{$\embed_3$}}
  \psfrag{E4}{{$\embed_4$}}

  \includegraphics[width=7cm]{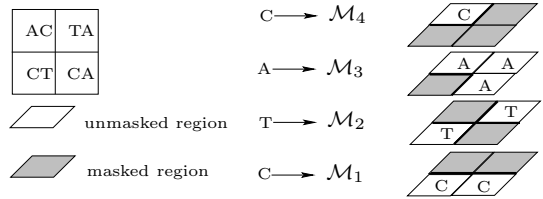}
  \label{fig:masksub1}
 }
   \caption{\small Asynchronous synthesis of a $2 \times 2$ microarray.
The deposition sequence $\mathcal{D} =$ CTAC
corresponds to %
four masks $\mask_1$, $\mask_2$, $\mask_3$,
and $\mask_4$.
The masked regions are shaded and the border between the masked and unmasked regions is
represented by bold lines.
}
  \label{fig:mask}
\end{center}
\end{figure}

The problem of finding both the placement and the embedding is termed the Border Minimization Problem (BMP). If the placement is given and the task is to find only the embedding, we speak of P-BMP. We refer the reader to Section~\ref{sec:prelim} for formal definitions of BMP and P-BMP.

\paragraph*{{\sc Variants of border minimization.}}
In this paper we consider the exhaustive variants of BMP and P-BMP, termed $\BMP$ and $\PBMP$ respectively.
The difference is that in $\PBMP$ (and, consequently, in $\BMP$)  we assume that a mask is always applied exhaustively (we call this the \emph{exhaustive rule}). More precisely, when a mask that synthesizes a character $c$ is applied,  the mask has a transparent cell wherever the corresponding sequence begins with the character~$c$.

Without this assumption it is possible to artificially increase the length of the deposition sequence which, as a consequence, also increases the length of the sequence of masks. 
In most application scenarios this is undesirable, since
applying a mask requires an additional cycle of work that causes a waste of material and can also introduce new errors. A second advantage of these exhaustive variants is that they allow the concise description of solutions: a solution to $\PBMP$ is fully characterized by the deposition sequence, while for P-BMP it is also necessary to explicitly describe each mask in the sequence. To clarify, we remark that an optimal exhaustive solution need not always be an optimal solution for P-BMP (or BMP): there are cases where the border length can increase.

We illustrate the usefulness of the assumption by a simple example. In the $\PBMP$ instance $a|b|a$, this assumption indeed helps to reduce the number of masks without increasing the border length. A non-exhaustive optimal solution might work on the left $a$ first, while an exhaustive optimal solution works on both $a$ concurrently. Even though the border length is in both cases $4$, the non-exhaustive case could require an additional mask.

\begin{table}[t!]\centering
	\begin{tabular}{l|l|l|l|l|l|}%
		& $c$ or $c,r$ & $c,\ell$ & $c,\ell,r$ & $c,o$ \\ 
		\hline 
		\hline
		\PBMP & \paraNP-h (Prop.~\ref{prop:pbmp-paraNP}) & FPT (Prop.~\ref{prop:PBMPcl}) & FPT (Prop.~\ref{prop:PBMPcl}) & FPT (Thm.~\ref{thm:PBMPco}) \\
		\BMP & \paraNP-h (Thm.~\ref{thm:bmp-paraNP}) & \multicolumn{1}{c|}{open} & FPT (Thm.~\ref{thm:BMPclr}) & FPT (Thm.~\ref{thm:BMPco}) \\
		\hline
	\end{tabular}
	\caption{Overview of results.}
  \label{tab:results}
\end{table}

\paragraph*{{\sc Our results.}}

Our results are summarized in Table~\ref{tab:results}.
In this paper we investigate the parameterized complexity of the $\BMP$ and $\PBMP$ problems under several natural parameters. First of all,  throughout this work we consider the number of available nucleotides $c$ (i.e., the alphabet size) as a parameter. Notice that this assumption does not impose a serious restriction, since in practice the number of available nucleotides is very limited (or even constant). 
Orthogonal to this assumption we explore the parameterized complexity of the $\BMP$ and $\PBMP$ problem with respect to three natural parameters, i.e., the maximum length of a sequence in the array ($\ell$), the maximum border length cost ($o$), and the maximum number of rows in the array ($r$). Since errors become more likely as the length of the sequence grows, the length of the constructed probes will be rather limited. Notice that the parameter $o$ models the cost of a solution and hence is also a natural parameter. Finally, with the maximum number of rows $r$ the shape of the array is restricted in the sense that the one dimension does not grow arbitrarily. This is, in particular, interesting because it allows to generalize from the one-dimensional case studied in~\cite{PopaWY12}.

More precisely, we show fpt-algorithms for $\BMP$ and $\PBMP$ if we are given either $c,\ell,r$ or $c,o$ as parameters. We complement these results with parameterized intractability results, i.e., by showing \paraNP-hardness. We use a polynomial time reduction from $\PBMP$ to $\BMP$ to build upon the result that $\PBMP$ parameterized by $c$ and $r$ is \paraNP-hard\footnote{Although in~\cite{PopaWY12} only \NP-hardness is proven for P-BMP, the reduction can also be used to show \paraNP-hardness for $\PBMP$ when parameterized by $c$ and $r$.} and obtain hereby \paraNP-hardness for $\BMP$ parameterized by $c$ and $r$. Notice that with the exception of $\BMP$ parameterized by $c$ and $\ell$, we obtain a full parameterized complexity map of the two considered problems with respect to all additional parameters considered in this paper. We furthermore provide a reduction relating the complexity of $\BMP$ parameterized  by $c$ and $\ell$ to $k$-\textsc{Balanced Partition} on grids, a well-studied problem whose parameterized complexity on grids is open (Proposition~\ref{prop:BMPclred}).

The rest of the paper is organized as follows. In Section~\ref{sec:prelim} we introduce the problems formally and give preliminaries. Then, in Section~\ref{sec:hardness} we show the reduction from $\PBMP$ to $\BMP$. Section~\ref{sec:fpt} introduces the fpt-algorithms and, finally, in Section~\ref{sec:conclusions} we present conclusions and open problems.

\section{Preliminaries}
\label{sec:prelim}

For $n\in\mathbb{N}$, we use $[n]$ to denote the set $\{1,\ldots,n\}$. For two sequences $s_1, s_2$, we use $s_1 \cdot s_2$ to mark their concatenation. 

The microarray has size $r \times m$, where $r$ is the number of rows and $m$ is the number of columns. The multiset of input sequences (also called \emph{probes}) is denoted by $\probeset = \{s_1, s_2,\ldots, s_{r\cdot m}\}$ and the input alphabet by $\Sigma$. Moreover, let $c = |\Sigma|$. For any sequence $s_i$,
we denote the length of the sequence by $\ell_i$ and the $t$-th character of a sequence $s_i$ by $s_i[t]$. We use $\ell$ for the maximum length of the probes, i.e., $\ell=\max_{i\in [r\cdot m]} \ell_i$. Two cells of the array $v_1=(x_1, y_1)$ and $v_2=(x_2, y_2)$ are said to be {\em neighbors} if $|x_1 - x_2| + |y_1 - y_2| = 1$. For each cell $v$, we denote the set of neighbors of $v$ by $\neighbor(v)$.

In order to give the formal definition of BMP, we introduce several notions related to the synthesis process.

\begin{definition}
A {\em placement} of the probe sequences is a
bijective function $\placement$
that maps each probe sequence
to a unique cell in the array.
\end{definition}

\begin{definition}
A {\em deposition sequence} $D$ for a set of sequences $\probeset$ is a sequence of characters which is a common supersequence of all sequences in $\probeset$.
\end{definition}

\begin{definition}
An {\em embedding} of a sequence $s_i$ into a deposition sequence $D$ is a length-$|D|$ sequence $\embed_i$ over alphabet $\Sigma\cup \{-\}$ such that:
\begin{enumerate}
\item $\embed_i$ contains precisely $|s_i|$ characters other than ``$-$'' occurring at positions $\embed_i[u_1],$ $\embed_i[u_2], \dots, \embed_i[u_{|s_i|}]$,
\item $u_1$ is the minimum position such that $\embed_i[u_1]=s_i[1]$,
\item for $2\leq j\leq |s_i|$, $u_j$ is the minimum position such that $\embed_i[u_j]=s_i[j]$ and $u_{j-1}<u_j$.
\end{enumerate}
\end{definition}

Informally, $\embed_i$ captures how a sequence is built (or, equivalently, deleted) by the deposition sequence; notice that due to the exhaustive rule, the embedding is uniquely determined by the deposition sequence.
An {\em embedding} of a set of probes $\probeset$ into a deposition sequence $D$ is then denoted by $\embed_D=\{\embed_1, \embed_2, \ldots, \embed_{|\probeset|} \}$. Note that we will drop the subscript when the associated deposition sequence is clear from the context. The final key notion we need are masks. 

\begin{definition}
A mask $\mask$ (for some character $c$) is a 2D-array such that $\mask(i,j)$ is either $c$ or a space ``$-$'' (here the space means that the character is not deposited into this cell). 
\end{definition}

The sequence of masks associated with a deposition sequence $D$ and a placement $\placement$ is $\omega=\mask_1, \dots, \mask_{|D|}$ where $\mask_i(a,b)=\embed_{\placement^{-1}(a,b)}[i]$ for $i\in [|D|]$. Notice that due to the exhaustive rule, a mask for character $c$ is always maximal with respect to $c$, i.e., there is no ``$-$'' in the mask that could be replaced by $c$. We introduce now the \emph{border length} of a given placement of the probes in the array, which is the value we aim to optimize.

\begin{definition}

Let $\border_{D}(s_i,s_j)$ be the Hamming distance between $\embed_i$ and $\embed_j$ (with respect to deposition sequence $D$).
The {\em border length} of a placement $\placement$
and a deposition sequence $D$ is then defined as
the sum of borders over all pairs of neighboring probe sequences
\begin{equation}
 \borderlen(\placement, D)=
   \displaystyle \sum_{\scriptsize
    \begin{array}{c}
    \forall i,j\in \mathbb{N}: i<j<|\probeset|\\
    \wedge~\placement(s_j) \in \neighbor(\placement(s_i))
    \end{array}
    } \border_{D}(s_i, s_j).
 \label{eq:bmpcost}
\end{equation}
\end{definition}

We can also equivalently define border length in terms of the border length of all the masks.
\begin{definition}

For any mask $\mask$ of deposition character $x$,
the border length of $\mask$, denoted by $\borderlen(\mask)$,
is defined as
the number of pairs of neighboring cells $(i_1,j_1)$ and $(i_2,j_2)$
such that $\mask(i_1,j_1)=x$ and $\mask(i_1,j_1) \not= \mask(i_2,j_2)$.
For a placement and deposition sequence that corresponds to
a sequence of masks $\mask_1$, $\mask_2$, $\cdots$, $\mask_{|D|}$, we let
\begin{equation}
 \borderlen(\placement, D)=
  \sum_{h=1}^{|D|} \borderlen(\mask_h)
 \label{eq:maskcost}
\end{equation}
\end{definition}

The $\BMP$ and the $\PBMP$ problem are defined as follows.
\begin{prob}
In the $\BMP$ problem, we are given $r,m\in \mathbb{N}$ and a multiset of $r\cdot m$ sequences $\probeset$. The objective is to find a placement $\placement$ and
a deposition sequence $D$ so that $\borderlen(\placement, D)$
is minimized.
\end{prob}

\begin{prob}
In the $\PBMP$ problem, we are given $r,m\in \mathbb{N}$ and a multiset of $r\cdot m$ sequences $\probeset$ and a placement $\placement$. The objective is to find a deposition sequence $D$ so that $\borderlen(\placement, D)$ is minimized.
\end{prob}

For a set $\pi\subseteq\{c,r,\ell,o\}$, we denote by $\BMP_\pi$ ($\PBMP_\pi$) the $\BMP$ ($\PBMP$) problem parameterized by $\pi$. For a problem $\BMP_\pi$ ($\PBMP_\pi$) where $o\in \pi$, we assume that an upper bound on the border length $o$ is additionally given in the input and only solutions with minimum border length $\leq o$ are admitted.

We conclude this section with some useful observations. A deposition sequence $D$ is called \emph{redundant} if it contains a character $D[i]$ such that $\embed_j[i]=``-$'' for each $\embed_j\in \embed$. Note that for any redundant deposition sequence $D$ and any placement $\placement$, it holds that $\borderlen(\placement,D)=\borderlen(\placement,D')$, where $D'$ is obtained by deleting the redundant character $D[i]$. We say that a deposition sequence $D$ is \emph{good} if it is not redundant.

\begin{observation}
\label{obs:good}
Let $(\placement, D)$ be such that $\borderlen(\placement, D)$ is minimized for some $(\probeset, r, m)$. If $D$ is redundant, then there exists a subsequence $D'$ of $D$ such that $\borderlen(\placement, D')=\borderlen(\placement, D)$ and $D'$ is good.
\end{observation}

As a consequence, when searching for optimal solutions of these problems it suffices to consider only good deposition sequences. Aside from the trivial (quadratic) algorithm for computing the border length for a fixed deposition sequence and placement, we will utilize another algorithm which will in some cases yield better running times:

\begin{proposition}
\label{prop:linear}
For any given $(\placement, D, \probeset, r, m)$, there exists an algorithm which computes $\borderlen(\placement, D)$ in time $\bigO{|\probeset|+p^2\cdot |D|}$, where $p$ is the number of distinct sequences in $\probeset$.
\end{proposition}

\begin{proof}
The algorithm proceeds in four steps. First, in time $O(|\probeset|)$ it finds all unique sequences in $\probeset$ and stores them in a set $Q$ along with a mapping $\eta:S\rightarrow Q$ which maps sequences from $S$ to their representative in $Q$. Second, in time $O(p^2\cdot |D|)$ it computes and stores $\border_D(q_1,q_2)$ for each  $q_1, q_2\in Q$. Third, in time $O(|\probeset|)$ for each sequence $s\in \probeset$ it computes the set $R_s=\placement^{-1}(\neighbor(\placement(s)))$ of neighboring sequences. Finally, in time $O(|\probeset|)$ it  computes $\frac{1}{2}\displaystyle \sum_{\scriptsize
    \forall s\in\probeset, r\in R_s
    } \border_{D}(\eta(s),\eta(r))$
    which is easily seen to be equal to $\borderlen(\placement,D)$. 
\end{proof}

\subsection{Parameterized Complexity}
Parameterized algorithmics is a promising approach to obtain efficient algorithms for fragments of computationally hard problems.
The aim is to find a parameter that describes the structure of the instance such that the combinatorial explosion can be confined to this parameter.
In a parameterized complexity analysis the runtime of an algorithm is studied with respect to the input size~$n$ and a parameter $k\in\nn$ (or a combination of parameters).
For a more detailed introduction we refer to the literature~\cite{DowneyFellows99,FlumGrohe06%
}.

Formally, a \textit{parameterized problem} is a subset of $\Sigma^*\times\nn$, where $\Sigma$ is the input alphabet. If a combination of parameters $k_1,\ldots, k_l$ is considered, the second component of an instance $(x,k)$ is given by $k=\sum_{1\leq i\leq l}k_i$.
The class \FPT (\textit{fixed-parameter tractable}) contains all problems that can be decided by an algorithm running in $f(k)\cdot n^{\bigO{1}}$ time, where $f$ is a computable function and $n$ is the input size.
Such algorithms are often called fixed-parameter tractable (fpt).  

Let $L_1$ and $L_2$ be parameterized problems, with $L_1\subseteq \Sigma_1^*\times\nn$ and $L_2\subseteq \Sigma_2^*\times\nn$. A \textit{parameterized reduction} (or fpt-reduction) from $L_1$ to $L_2$ is a mapping $P:\Sigma_1^*\times\nn\rightarrow\Sigma_2^*\times\nn$ such that \begin{inparaenum}[(1)]   \item $(x,k)\in L_1$ iff $P(x,k)\in L_2$;   \item the mapping can be computed by an fpt-algorithm with respect to parameter $k$;   \item there is a computable function $g$ such that $k'\leq g(k)$, where $(x',k')=P(x,k)$. \end{inparaenum}

There is a variety of classes capturing \textit{parameterized intractability}.
For our results, we require only the class \paraNP~\cite{FlumGrohe03}, which is defined as the class of problems that are solvable by a nondeterministic Turing-machine in fpt-time.
We will make use of the characterization of \paraNP-hardness given by Flum and Grohe~\cite{FlumGrohe06}, 
Theorem 2.14: any parameterized problem that remains \NP-hard when the parameter is set to some constant is \paraNP-hard.
Showing \paraNP-hardness for a problem rules out the existence of an fpt-algorithm under the usual complexity theoretic assumptions.

\section{Hardness}
\label{sec:hardness}

In this section we overview and present new (parameterized) intractability results for $\BMP$ and $\PBMP$ with respect to several combinations of parameters. As our starting point, we notice that the \NP-hardness proof for P-BMP of Popa, Wong and Yung \cite{PopaWY12} can be straightforwardly adapted to $\PBMP_{c,r}$.

\begin{proposition}[{cf.~\cite[Theorem~1]{PopaWY12}}]\label{prop:pbmp-paraNP}
  $\PBMP_{c,r}$ is \paraNP-hard.
\end{proposition}

\begin{proof}
Observe that the reduction used in the proof of Theorem~1 in~\cite{PopaWY12} constructs instances of BMP which only contain $3$ characters. Furthermore, while the instances are formally defined as square arrays, all rows below the $5$-th contain only a dummy character \$ and hence can be omitted without loss of generality. Finally, by Lemma~2 in~\cite{PopaWY12} it follows that optimal exhaustive solutions for these BMP instances are also optimal solutions (in fact, it is these exhaustive solutions that are used to prove Theorem~1 in~\cite{PopaWY12}).
\end{proof}

The hardness result for $\BMP$ relies on a new polynomial-time reduction from $\PBMP$ to $\BMP$.
We believe that this reduction is an interesting result on its own, as it is one of the first results that relates the complexity of these two problems in a general setting. We begin by showcasing a tool for forcibly ``separating'' any optimal deposition sequence.

\begin{lemma}
\label{lem:sep}
Let $\cI=(\probeset,r,m)$ be an instance of $\BMP$ such that each $s\in \probeset$ consists of a prefix $s_{\textit{pre}}\in{\Sigma_{\textit{pre}}^*}$, a fixed separator $\textit{sep}\in{(x^*y^*)^*}$ and a suffix $s_{\textit{suf}}\in{\Sigma_{\textit{suf}}^*}$, where $\Sigma_{\textit{pre}}, \Sigma_{\textit{suf}}, \{x,y\}$ form a partition of $\Sigma$. Let $u\geq 8\cdot \textit{max}_{s\in \probeset}(|s_{\textit{pre}}|)+8\cdot \textit{max}_{s\in \probeset}(|s_{\textit{suf}}|)+1$. If $\textit{sep}=(x^{r\cdot m\cdot u}\cdot y^{r\cdot m\cdot u})^{r\cdot m\cdot u}$ then every optimal good deposition sequence has the form $D_{\textit{pre}}\cdot \textit{sep}\cdot D_{\textit{suf}}$ where $D_{\textit{pre}}\in{\Sigma_{\textit{pre}}^*}$ and $D_{\textit{suf}}\in{\Sigma_{\textit{suf}}^*}$.
\end{lemma}

\begin{proof}
Notice that $r\cdot m\cdot u-1$ forms a trivial upper-bound on the border length of $\cI$, as witnessed by any deposition sequence of the form $D_{\textit{pre}}\cdot \textit{sep}\cdot D_{\textit{suf}}$ (regardless of placement). Indeed, there are at most $4r\cdot m$ pairs of neighboring cells in the array, and for each such pair the border length is bounded by the hamming distance between the embeddings placed on these cells, where any deposition sequence of this form yields a bound of $2\cdot \textit{max}_{s\in \probeset}(|s_{\textit{pre}}|)+2\cdot \textit{max}_{s\in \probeset}(|s_{\textit{suf}}|)$.

Consider any optimal good deposition sequence $D$ and let $p\in \Sigma_{\textit{pre}}, q\in \Sigma_{\textit{suf}}$. Consider for a contradiction that $qp$ is a subsequence of $D$. Then $\textit{pre}\cdot qp$ would also be a subsequence of $D$; however, each mask for a character in $\textit{pre}$ would yield an increase of the border length by at least $1$, since the array contains a cell in the array where this mask cannot be applied (specifically, this is the cell containing the sequence beginning with $p$). This would already break the upper-bound provided above. hence $qp$ cannot be subsequence of $D$.

Next, consider for a contradiction that $qy$ is a subsequence of $D$. Then $\textit{sep}\cdot qy=(x^{r\cdot m\cdot u}y^{r\cdot m\cdot u})^{r\cdot m\cdot u}\cdot qy$ would also be a subsequence of $D$. This means that there exist two embeddings $\embed_1, \embed_2$ which differ in the positions of their first, second, third,\dots,$({r\cdot m\cdot u})^2$-th $y$ characters. Let $\textit{offset}_x$ be the number of masks for $x$ which occur between the position of the first $y$ character in $\embed_1$ and the first $y$ character in $\embed_2$; notice that $0<\textit{offset}_x\leq r\cdot m\cdot u$. Each mask for $x$ in the offset has a border length of at least $1$, since there is a sequence $s_2$ in the array which begins with $y$.
If $\textit{offset}_x<r\cdot m\cdot u$ then the upper-bounded on the border length of $D$ is broken by the fact that that $x^{r\cdot m\cdot u}y^{r\cdot m\cdot u}$ occurs $(r\cdot m\cdot u)$-many times in succession in the deposition sequence, and each occurrence would necessarily increase the border length by at least $1$. On the other hand, if $\textit{offset}_x=r\cdot m\cdot u$ then the upper-bound on the border length would be broken already by all the masks for $x$ which occur in the offset.

By a symmetric argument, we obtain that $xp$ also cannot occur as a subsequence of $D$. Hence the deposition sequence must have the form $D_{\textit{pre}}\cdot \textit{sep}\cdot D_{\textit{suf}}$.
\end{proof}

Observe that ``flipping'' the array horizontally or vertically preserves the optimal border length but formally changes the placement $\placement$. The purpose of the following key lemma is to provide a tool to fix the optimal positions of probes in the array; to this end, we will be considering placements which are unique up to these simple symmetries.

\begin{lemma}\label{lem:ab-fix-optimal}
  Let $a,b,x,y \in \Sigma$ and $r,m,t\in\mathbb{N}$.
  Consider an $r\times m$ array, and probes $\probeset =\{a^{i\cdot t}\cdot sep \cdot b^{j\cdot t}\mid i\in [r]\text{ and }j\in [m]\}$. Then:
  \begin{enumerate}
  \item the unique optimal placement $\placement_0$ (up to simple symmetries) places each probe $a^{i\cdot t}\cdot \textit{sep}\cdot b^{j\cdot t}$ in cell $(i,j)$, 
  \item the unique optimal good deposition sequence is $D_0=a^{r\cdot t}\cdot \textit{sep}\cdot b^{m\cdot t}$, and 
  \item for any placement $\placement\neq \placement_0$ (except for symmetries of $\placement_0$) and any deposition sequence $D$, it holds that  $\borderlen(\placement,D)\geq \borderlen(\placement_0,D_0)+t$.
  \end{enumerate}
\end{lemma}

\begin{proof}

We proceed in two steps. First, we compute the border length of $(\placement_0,D_0)$. Then, we establish that $\placement_0$ is the only optimal placement up to the above-mentioned simple symmetries, and that other placements yield a border length which is lower-bounded by $t+\borderlen(\placement_0,D_0)$. Notice that $D_0$ is the only optimal good deposition sequence regardless of placement by Lemma \ref{lem:sep}.

\begin{claim}
$\borderlen(\placement_0, D_0) = ((r-1)\cdot m+r\cdot(m-1))\cdot t$.
\end{claim}
\begin{cproof}
For character $a$, we start with $t$-many masks that contain character $a$ in each cell. Notice that these masks have border length zero. Then we continue with $t$-many masks that have character ``$-$''  in the first row and character $a$ everywhere else. Each of these masks has border length $m$. Next we use $t$-many masks, where the first two rows contain character ``$-$'', and so on. In total, we obtain a border length of $(r-1)\cdot m\cdot t$ for character $a$. For character $x$ and $y$, all masks contain character $x$ or $y$ in each cell and hence all have a border length of zero. Finally, for character $b$ the procedure is analogous -- we simply swap columns and rows. This gives a border length of $r\cdot (m-1)\cdot t$ for character $b$.
\end{cproof}

Now consider any optimal solution $(\placement,D)$. The fact that $D=D_0$ follows from Lemma \ref{lem:sep}. We now proceed to the core of our proof. Notice that for each pair of probes $s_1,s_2\in \probeset$ it holds that $\border_{D_0}(s_1,s_j)\geq t$. We say that $s_1,s_2$ are \emph{similar} if $\border_{D_0}(s_1,s_j)=t$. Since the number of pairs of cells which are neighbors in an $r\times m$ array is exactly $(r-1)\cdot m+r\cdot (m-1)$ and $\borderlen(\placement_0, D_0) = ((r-1)\cdot m+r\cdot(m-1))\cdot t$, any optimal placement $\placement$ may only place probes which are similar into neighboring cells. Furthermore, if a placement $\placement$ is not optimal, then $\borderlen(\placement,D)\geq t+\borderlen(\placement_0,D_0)$ since for any $s_1, s_2$ which are not similar it holds that $\border_{D_0}(s_1,s_j)\geq 2t$.

Let us denote the cells which have at most $3$ neighbors in the array the \emph{perimeter} and the cells which have at most $2$ neighbors the \emph{corners}.
For the final part of the proof, we use the inductive assumption that $\placement_0$ is the unique optimal placement for all $r'\times m'$ arrays such that $r'<r$ and $m'<m$ as long as the placement of at least two corners is fixed. Furthermore, we assume that $\textit{min}(r,m)>1$; the lemma trivially holds for $\textit{min}(r,m)=0$, and is easily seen that $\textit{min}(r,m)=1$ the optimal placement must be an ascending sequence, which is unique if its corners/endpoints are fixed. 

For each $s\in \probeset$, let $\textit{sim}(s)$ denote the set of probes which are similar to $s$. Notice that there are precisely four probes such that $|\textit{sim}(s)|=2$ and precisely $2r+2m-4$ probes such that $|\textit{sim}(s)|=3$, and there is a unique (up to symmetry) placement of these probes in the corners and perimeter so that similar probes are placed on neighboring cells (see Fig. \ref{fig:perimeter}). Let $\probeset_0$ contain all the probes placed into the perimeter.

\begin{figure}
\centering
\definecolor{Gray}{gray}{0.85}
\newcolumntype{g}{>{\columncolor{Gray}}c}
\begin{tabular}{ | g | c | c | c | c | c | g | }
\hline
  \rowcolor{Gray}$s_{1,1}$ & $s_{1,2}$ & $s_{1,3}$ & $\cdots$ & $s_{1,m-2}$ & $s_{1,m-1}$ & $s_{1,m}$ \\ \hline
  $s_{2,1}$ & $s_{2,2}$ & $s_{2,3}$ & $\cdots$ & $s_{2,m-2}$ & $s_{2,m-1}$ & $s_{2,m}$ \\ \hline
  $\vdots$ & $\vdots$ & $\vdots$ & $\cdots$ & $\vdots$ & $\vdots$ & $\vdots$ \\ \hline
  $s_{r-1,1}$ & $s_{r-1,2}$ & $s_{r-1,3}$ & $\cdots$ & $s_{r-1,m-2}$ & $s_{r-1,m-1}$ & $s_{r-1,m}$ \\ \hline
  \rowcolor{Gray}$s_{r,1}$ & $s_{r,2}$ & $s_{r,3}$ & $\cdots$ & $s_{r,m-2}$ & $s_{r,m-1}$ & $s_{r,m}$ \\ \hline
\end{tabular}
\caption{An $r\times m$ array. The corners and the perimeter are highlighted in gray.}
\label{fig:perimeter}
\end{figure}

Notice that the placement of these probes on the perimeter precisely matches $\placement_0$, and the placement of probes such that $|\textit{sim}(s)|\leq 2$ in $\probeset'=\probeset-\probeset_0$ is fixed by the placement of $\probeset_0$ in the perimeter. 

If $\textit{min}(r,m)=2$ then this concludes the proof. If $\textit{min}(r,m)=3$ then the remaining placement reduces to the placement of $\probeset'=\probeset-\probeset_0$ into a one-dimensional array, which is unique when the corners are fixed. Finally, if $\textit{min}(r,m)=4$ then the remaining placement reduces to the placement of $\probeset'$ into an $(r-2)\times (m-2)$ array, which is again unique by our inductive hypothesis.
\end{proof}

With Proposition \ref{prop:pbmp-paraNP} and Lemma \ref{lem:ab-fix-optimal}, we can proceed to:

\begin{theorem}\label{thm:bmp-paraNP}
  $\BMP_{c,r}$ is \paraNP-hard.
\end{theorem}

\begin{proof}
We provide a reduction from $\PBMP_{c,r}$, which is \paraNP-hard by Proposition \ref{prop:pbmp-paraNP}. Let $\Sigma'$ be the language of $\PBMP_{c,r}$, $x_1,y_1,x_2,y_2\not\in \Sigma'$ and $\Sigma=\Sigma'+\{x_1,y_1,x_2,y_2\}$. From any instance $\cI'=(\probeset',\placement',r,m)$ of $\PBMP_{c,r}$, we construct an instance $\cI=(\probeset,r,m)$ of $\BMP_{c,r}$ as follows. For each $s\in \probeset'$ such that $\placement'(s)=(i,j)$ we put $a^{i\cdot t}\cdot sep_1\cdot b^{j\cdot t}\cdot sep_2\cdot s$ into $\probeset$, where:
\begin{itemize}
\item $t>(\textit{max}_{s\in \probeset'}|s| \cdot r\cdot m)^2$.
\item $\textit{sep}_1=(x_1^{r\cdot m\cdot u_1}y_1^{r\cdot m\cdot u_1})^{r\cdot m\cdot u_1}$
\item $\textit{sep}_2=(x_2^{r\cdot m\cdot u_2}y_2^{r\cdot m\cdot u_2})^{r\cdot m\cdot u_2}$
\item the constants $u_1, u_2$ for $\textit{sep}_1$ and $\textit{sep}_2$ respectively are sufficiently large so as to satisfy the condition of Lemma \ref{lem:sep}; for instance, $u_2>100t^3$ and $u_1>1000t^4$.
\end{itemize}

By Lemma \ref{lem:sep} we have that any optimal good deposition sequence for $\cI$ must have the form $a^{r\cdot u}\cdot \textit{sep}_1\cdot b^{m\cdot u}\cdot \textit{sep}_2\cdot D'$. Let us now compare an arbitrary solution $(\placement,D)$ to $(\placement',D)$. By Lemma \ref{lem:ab-fix-optimal}, either $\placement$ is equivalent to $\placement'$ by symmetry, or the border length of masks for $a,x_1,y_1,b$ in $(\placement,D)$ will be at least $t$ greater than the border length of these masks in $(\placement',D)$. However, $t$ was chosen to be sufficiently large to exceed the worst-case border length of all masks for $\Sigma'$. So we conclude that any optimal solution for $\cI$ must use a placement which is either the same as or symmetric to $\placement'$.

Finally, observe that after the last mask of $\textit{sep}_2$ is applied, the remainder of $\cI$ is equivalent to $\cI'$, and hence $D'$ is also a solution to $\cI'$.
\end{proof}

Theorem \ref{thm:bmp-paraNP} and Proposition \ref{prop:pbmp-paraNP} show that one cannot hope to find an fpt-algorithm for $\BMP$ or $\PBMP$ parameterized by any subset of $\{c,r\}$.
These results complete the hardness part of our complexity map for $\BMP$ or $\PBMP$. %
For $\BMP_{c,\ell}$ it remains open whether the problem is fixed parameter tractable. Still, we can relate this problem to $k$-\textsc{Balanced Partition}, a problem studied well in the literature~\cite{AndreevR06,FeldmannThesis12,Feldmann13}.

In a $k$-\textsc{Balanced Partition} instance we are given a graph $G=(V,E)$ with $|V|=n$. The question is to find a partition of the vertices $V$ into $k$ sets $V_1,\ldots, V_k$ such that $|V_i|\leq\ceil{\frac{n}{k}}$ for all $1\leq i\leq k$, and the cut size (i.e., the number of edges $\{x,y\}$ such that $x\in V_i$, $y\in V_j$, and $i\neq j$) is minimized. We remark that, to the best of our knowledge, the parameterized complexity of $k$-\textsc{Balanced Partition} parameterized by $k$ is open on solid rectangular grids \cite{FeldmannThesis12}. Below we show that $k$-\textsc{Balanced Partition} on solid rectangular grids can be reduced to $\BMP$ and hence $\BMP$ is at least as hard as $k$-\textsc{Balanced Partition}.

\begin{proposition}
\label{prop:BMPclred}
  There is a polynomial time reduction from $k$-\textsc{Balanced Partition} on solid rectangular grids to $\BMP$.
\end{proposition}

\begin{proof}
  Let $G=(V,E)$ be a solid rectangular grid of size $r\times m$ with $|V|=n$.
  Further, let $l,x\in\mathbb{N}_0$ such that $n=l\cdot k+x$ and $\Sigma=\{c_1,\ldots,c_k\}$.
  We construct a probe set $\probeset = \{\text{$f(i)$ copies of sequence $c_i$} \mid i\in[l] \}$, where function $f(i)=l+1$ if $1\leq i\leq x$ and $f(i)=l$ otherwise.
  It is easy to verify that the placement $\placement$ of a $\BMP$ solution gives also a solution to $k$-\textsc{Balanced Partition} if the characters in $\Sigma$ are seen as partition sets $V_1,\ldots V_k$.
\end{proof}

\section{Fpt-Algorithms}
\label{sec:fpt}

In the following sections we discuss fpt-algorithms for several parameters. The first group focuses on sequences of moderate length and an array whose size is primarily growing in one dimension, i.e., on the parameters $c$, $\ell$, and $r$. In contrast, the second group parameterizes by $c$ and the maximum admissible border length $o$.

\subsection{Fpt-Algorithm for $\PBMP_{c,\ell}$}
\label{sub:fptcl}

Our first algorithm provides a basic introduction to the techniques used later on.

\begin{observation}
\label{obs:bounds}
For any instance $(\probeset, r, m)$ of $\BMP_{c,\ell}$, there are at most $c^{\ell}$ unique sequences in $\probeset$.
\end{observation}

\begin{lemma}
\label{lem:bounds}
For any instance $(\probeset, r, m)$ of $\BMP_{c,\ell}$ or any instance $(\probeset, \placement, r, m)$ of $\PBMP_{c,\ell}$ it holds that $|D|\leq c^{\ell}\cdot \ell$ for any good deposition sequence $D$.
\end{lemma}

\begin{proof}
Assume towards contradiction that there is a good deposition sequence $D$ which contains $|D|>c^{\ell}\cdot \ell$ characters. Since the total number of distinct sequences $s_i\in \probeset$ is bounded by $c^{\ell}$, the total number of distinct embeddings $\embed_i$ is also bounded by $c^\ell$. Each embedding $\embed_i$ contains at most $\ell$ characters in $\Sigma \backslash \{-\}$. Hence by the pigeon-hole principle there must exist some $j\in [|D|]$ such that $\embed_i[j]=``-$'' for all $i\in [|\probeset|]$, which implies that $D$ is not good (contradiction).
\end{proof}

At this point we can already prove:

\begin{proposition}
\label{prop:PBMPcl}
$\PBMP_{c,\ell}$ is fixed parameter tractable, and there exists an algorithm for $\PBMP_{c,\ell}$ which runs in time $c^{c^{\bigO{\ell}}}|\probeset|$.
\end{proposition}

\begin{proof}
By Lemma~\ref{lem:bounds}, it suffices to search for deposition sequences of length at most $c^{\ell}\cdot \ell$. We loop through all of the at most $c^{c^{\ell}\cdot \ell}$ such deposition sequences, and for each sequence $D$ we compute $\borderlen(\placement, D)$ in time $O(|\probeset|+p^2\cdot |D|)$ by Proposition~\ref{prop:linear}. By Observation~\ref{obs:bounds} and Lemma~\ref{lem:bounds}, we obtain that $O(|\probeset|+p^2\cdot |D|)=O(|\probeset|+c^{3\ell}\ell)$, which altogether yields the runtime bound of $c^{c^{O(\ell)}}|\probeset|$.
\end{proof}

\subsection{Fpt-Algorithm for $\BMP_{c,\ell,r}$}
\label{sub:fptbmpclr}

We first introduce some notation for our arrays. Given an $r\times m$ array $A$, a \emph{column} is an $r\times 1$ sub-array of $A$. 
A \emph{column placement} into a column of $A$ is a mapping $\placement: [r]\rightarrow \probeset$ from the cells of $A$ to the multiset of probes.

\begin{observation}
\label{obs:columntypes}
For any instance $(\probeset, r, m)$ of $\BMP$, it holds that there are at most $c^{\ell\cdot r}$ distinct column placements.
\end{observation}

Hence for any fixed $r$ and $\probeset$, we can enumerate all possible column placements as $\placement_1,\placement_2,\dots, \placement_{c^{\ell\cdot r}}$.
Observe that, for any two column placements $\placement_t, \placement_{t'}$, it holds that either
\begin{inparaenum}[(i)]
\item $t=t'$ and $\placement_t(x)=\placement_{t'}(x)$ for all $x\in [r]$, or
\item $t\neq t'$ and $\placement_t(x)\neq\placement_{t'}(x)$ for at least one $x\in [r]$.
\end{inparaenum}

Any placement $\placement:s\in \probeset \mapsto (a\in \mathbb{N}, b\in \mathbb{N})$ into $A$ can be uniquely decomposed into a sequence of column placements $(\placement_{i(1)}, \placement_{i(2)},\dots \placement_{i(m)})$ 
where $\placement_{i(x)}(y)=\placement(x,y)$ and $i:[m]\rightarrow [c^{\ell\cdot r}]$.
The column placement $\placement_{i(j)}$ with $j\in [m]$ denotes that the $j$-th column of $A$ is of placement $i(j)$.
Furthermore, since $\placement$ is closed under permutation of non-distinct sequences in $\probeset$, each column placement can be uniquely identified by an $r$-tuple of sequences from $\probeset$, formally $\placement_{i(x)}=(s_1,s_2,\dots,s_r) \iff \placement_{i(x)}(y)=s_y$ for all $y\in [r]$. 

Next, we prove that when searching for optimal solutions for $\BMP$ it suffices to restrict ourselves to placements such that identical column placements appear in ``consecutive blocks''.

\begin{lemma}
\label{lem:consecutive}
Let $(\probeset, r, m)$ be an instance of $\BMP$, $D$ be a deposition sequence and $\placement$ be a placement which decomposes into $(\placement_{i(1)},\placement_{i(2)},\dots\placement_{i(m)})$. 
Then if there exist $a,b\in [m],~a+1<b,$ such that $\placement_{i(a)}=\placement_{i(b)}$ but $\placement_{i(a+1)}\neq \placement_{i(b)}$, then $\borderlen(\placement, D)\geq  \borderlen(\placement', D)$, where $\placement'$ decomposes into 
$$(\placement_{i(1)},\dots \placement_{i(a)},\placement_{i(b)}, \placement_{i(a+1)},\placement_{i(a+2)},\dots,\placement_{i(b-1)}, \placement_{i(b+1)}, \dots,\placement_{i(m)}).$$
\end{lemma}

\begin{proof}
Recall that by Equation~\ref{eq:bmpcost}, $\borderlen(\placement, D)$ is equal to the sum of Hamming distances of embeddings $\border_{D}(s_p,s_q)$ between neighboring $s_p,s_q\in \probeset$. Since the embeddings, and hence also the Hamming distances, are the same for $\borderlen(\placement, D)$ as for $\borderlen(\placement', D)$, the only difference between these values may arise from which sequences are neighbors. 

We say that two neighboring cells $v_1=(x_1, y_1)$ and $v_2=(x_2, y_2)$
are $x$-neighbors if $|x_1 - x_2|=1$ and $y$-neighbors otherwise, i.e., if $|y_1 - y_2|=1$; let $\neighbor_x(v)$ and $\neighbor_y(v)$ contain the $x$-neighbors and $y$-neighbors of $v$, respectively.
Notice that $y$-neighborhoods are identical between $\placement$ and $\placement'$, since the latter is obtained by permuting whole columns of the former. On the other hand, consider the difference between $x$-neighboring sequences in $\placement$ and $\placement'$. Notice that $\placement'$ is obtained by a simple permutation of the column placements of $\placement$ and in particular these differ only in the borders between $\{\placement_{i(a)},\placement_{i(a+1)},\placement_{i(b-1)},\placement_{i(b)},\placement_{i(b+1)}\}$. For convenience, we use $\bd$ to denote the total ``horizontal'' border between two column placements; formally:
$$\bd(u,t)=\displaystyle \sum_{\scriptsize
    \begin{array}{c}
    \forall x\in[r]
    \end{array}
    } \border_{D}(\placement_{i(u)}(x),\placement_{i(t)}(x)).$$
    
Now we can express the difference between the border lengths of both placements as $\borderlen(\placement',D)=\borderlen(\placement,D)
+\bd(a,b)+\bd(b,a+1)+\bd(b-1,b+1)-\bd(a,a+1)-\bd(b-1,b)-\bd(b,b+1)$. Since $\placement_{i(a)}=\placement_{i(b)}$, it holds that $\bd(a,b)=0$ and $\bd(b,a+1)=\bd(a,a+1)$. Furthermore, since the triangle inequality holds for Hamming distances (and $\border_D$ is defined as a Hamming distance between two sequences), we obtain $\bd(b-1,b+1)-\bd(b-1,b)-\bd(b,b+1)\leq 0$. Hence we conclude that $\borderlen(\placement',D)\leq \borderlen(\placement,D)$.
\end{proof}

We say that a placement $\placement$ is \emph{consecutive} if it decomposes into column placements $(\placement_{i(1)},\placement_{i(2)},\dots\placement_{i(m)})$ where for each $\placement_{i(a)}, \placement_{i(b)}$ such that $\placement_{i(a)}= \placement_{i(b)}$ and $a<b$ it holds that $\placement_{i(a)}= \placement_{i(c)}$ for all $a<c<b$.

\begin{corollary}
\label{cor:consecutive}
For any $\BMP$ instance $(\probeset, r, m)$, there exists an optimal solution $(\placement, D)$ such that $\placement$ is consecutive.
\end{corollary}

\begin{proof}
Let $(\placement', D)$ be a solution for $(\probeset, r, m)$. We can repeatedly apply Lemma~\ref{lem:consecutive} until we obtain a consecutive placement---notice that the number of times Lemma~\ref{lem:consecutive} can be applied is bounded by $m$.
\end{proof}

The next algorithm uses an Integer Linear Programming (ILP) subroutine.
ILP is a well-known framework for formulating problems and a powerful tool for the development of fpt-algorithms for optimization problems.
In following we only give a brief overview of the framework before we present the algorithm.

\begin{definition}[$p$-Variable Integer Linear Programming Optimization] Let $A\in \mathbb{Z}^{q\times p}, b\in \mathbb{Z}^{q\times 1}$ and $c\in \mathbb{Z}^{1\times p}$. The task is to find a vector $x\in \mathbb{Z}^{p\times 1}$ which minimizes the objective function $c\times \bar x$ and satisfies all $q$ inequalities given by $A$ and $b$, specifically satisfies $A\cdot \bar x\geq b$. The number of variables $p$ is the parameter.
\end{definition}

Lenstra \cite{Lenstra83} showed that \textsc{$p$-ILP}, together with its optimization
variant \textsc{$p$-OPT-ILP} (defined above), are in FPT. His running time was
subsequently improved by Kannan \cite{Kannan87} and Frank and Tardos
\cite{FrankTardos87} (see also \cite{FellowsLokshtanovMisraRS08}).

\begin{theorem}[\cite{FellowsLokshtanovMisraRS08,FrankTardos87,Kannan87,Lenstra83}]
\label{thm:pilp}
\textsc{$p$-OPT-ILP} can be solved using $\bigO{p^{2.5p+o(p)}\cdot L}$ arithmetic operations in space polynomial in $L$, $L$ being the number of bits in the input.
\end{theorem}

We are now ready to prove the main theorem of this subsection. 
\begin{theorem}
\label{thm:BMPclr}
$\BMP_{c,\ell,r}$ is fixed parameter tractable, and there exists an algorithm for $\BMP_{c,\ell,r}$ which runs in time $c^{c^{\bigO{\ell\cdot r}}}\cdot |\probeset|$.
\end{theorem}

\begin{proof}
We give a multi-step algorithm for $\BMP_{c,\ell,r}$:

 \begin{enumerate}
 \item \label{step1}We branch on the choice of deposition sequence $D$. By Observation~\ref{obs:good}, it suffices to consider only good deposition sequences, and by Lemma~\ref{lem:bounds} the number of good deposition sequences is bounded by $c^{c^{O(\ell)}}$.
 \item In view of Corollary~\ref{cor:consecutive}, we branch on which column placements appear in $\placement$ and the order in which they appear. Formally, we construct the set of all distinct column placements $\cT=\{\placement_1,\dots\}$, branch on all nonempty subsets $\cT'\subseteq \cT$. We then branch on all mappings $f:[t]\rightarrow [|\cT|]$ where $t=|\cT'|$. Since $|\cT|\leq c^{\ell\cdot r}$ by Observation~\ref{obs:columntypes}, there are at most $O(c^{c^{O(\ell\cdot r)}})$ choices of $f$. 
 \end{enumerate}
 
 For each fixed $f$, we hence obtain a \emph{template} $Q_f=(\placement_{f(1)},\placement_{f(2)},\dots,\placement_{f(t)})$. A consecutive placement $\placement$ \emph{matches} a template $Q_f$ if there exists a \emph{multiplicity function} $h:t\rightarrow \mathbb{N}$ such that $\placement$ decomposes into $(h(1)\cdot \placement_{f(1)},h(2)\cdot \placement_{f(2)},\dots,h(t)\cdot \placement_{f(t)})$ where $x\cdot \placement_z$ is shorthand for $x$ consecutive copies of $\placement_z$.
 
 \begin{enumerate}
 \item[3.] We compute the following constants:
 \begin{itemize}
 \item For each column placement $\placement_i=(s_1,s_2,\dots,s_r)\in \cT'$ we compute the total cost of its ``vertical borders'' $\bd^{vert}_i$ as follows: $$\bd^{vert}_i=\displaystyle \sum_{\scriptsize
     \begin{array}{c}
     \forall z\in[r-1]
     \end{array}
     } \border_{D}(s_z, s_{z+1}).$$
 \item We also compute the total ``horizontal cost'', which depends only on $D$ and $Q_f$ (since identical column placements do not have horizontal borders), as follows: $$\horizontalcost=\displaystyle \sum_{\scriptsize
     \begin{array}{c}
     \forall z\in[r], w\in [t-1] 
     \end{array}
     } \border_{D}(\placement^{-1}_{f(w)}(z),\placement^{-1}_{f(w+1)}(z)).$$
 \item For each distinct $s\in \probeset$ let $\#_s$ contain the number of occurrences of $s$ in $\probeset$.
 \item For each distinct $s\in \probeset$ and $\placement_i$ let $\#_s^i$ contain the number of occurrences of $s$ in $\placement_i$.
 \end{itemize}
 \item[4.] We construct and solve an \textsc{$p$-OPT-ILP} instance $\cI$ to compute the multiplicity function $h$ which contains the ``vertical cost'' variable $\verticalcost$, the variables $h(1), \dots, h(t)$ and the following constraints:
 \begin{enumerate}
 \item[a)] For each distinct $s\in \probeset$: $\#_s=\displaystyle \sum_{\scriptsize
     \begin{array}{c}
     \forall z\in[t]
     \end{array}
     } h(z)\cdot \#_s^z$.
 \item[b)] $\forall z\in [t]:h(z)>0$.
 \item[c)] $\verticalcost=\displaystyle \sum_{\scriptsize
     \begin{array}{c}
     \forall z\in[t]
     \end{array}
     } h(z)\cdot \bd^{vert}_z$.
 \item[d)] Minimize $\verticalcost$.
 \end{enumerate}
 The intuition of the constraints is as follows.
 Constraints of type a) ensure that the choice of multiplicities does not introduce too many/too few occurrences of some probe $s$ in the array.
 By the constraints of type b) it is ensured that the multiplicities are strictly positive.
 With help of constraint c) the vertical border cost for a certain choice of multiplicities is computed, which is in turn minimized by constraint d).
 
 \item[5.] Finally, for each choice of $D$, $\cT'$ and $f$ we store $\verticalcost+\horizontalcost$ and the table of values $h=(h(1),\dots,h(t'))$ from the optimal solution of $\cI$. After the branching is complete, we choose an arbitrary branch with minimum $\verticalcost+\horizontalcost$ and read the values $D, f, h$ associated with this branch. The algorithm then outputs $(\placement,D)$ where $\placement$ is computed from the template $Q_f$ given by $f$ and the multiplicity function given by $h$.
 \end{enumerate}
 
 \paragraph{Running time.} The number of branches processed after Step $1$ and Step $2$ is bounded by $c^{c^{O(\ell)}}2^{c^{\ell \cdot r}}c^{c^{O(\ell\cdot r)}}=c^{c^{O(\ell\cdot r)}}$ and this branching can be initialized in $O(|\probeset|)$ time. Step $3$ and the construction of $\cI$ can both also be completed in linear time, assuming multisets are implemented via a multiplicity function. $\cI$ contains $t\leq c^{\ell\cdot r}$ variables and has size linear in $\probeset$, and can thus be solved in time at most $c^{c^{O(\ell\cdot r)}}\cdot |\probeset|$ by Theorem~\ref{thm:pilp}. The time required to process Step $5$ is easily seen to be dominated by Step $1$ and $4$.
 
 \paragraph{Correctness.} Assume for a contradiction that the algorithm outputs $(\placement,D)$ but there exists an optimal solution $(\placement',D')$ such that $\borderlen(\placement',D')<\borderlen(\placement,D)$. Consider the template $Q'_f$ and multiplicity function $h'$ associated with $\placement'$. During the computation of our algorithm, the branch of $Q'_f$ and $D'$ had correctly computed the $\horizontalcost'$ component of $\borderlen(\placement',D')$. Furthermore, since $(\placement',D')$ is optimal, we obtain that $h'$ must be an optimal solution for the \textsc{$p$-OPT-ILP} instance $\cI'$ constructed for this branch; let $\verticalcost'$ be the output of $\cI'$. Then $\borderlen(\placement',D')=\verticalcost'+\horizontalcost'$ implies that $\verticalcost'+\horizontalcost'<\verticalcost+\horizontalcost$, which contradicts the assumed choice of branch $D$ and $Q_f$ in Step $5.$
 \end{proof}

\subsection{Fpt-Algorithm for $\PBMP_{c,o}$}
\label{sub:fptpbmpco}

Given an $r\times m$ array, a mask $\mask$ %
is called \emph{trivial} if $\mask(i,j)\neq\text{``$-$''}$ for all $i\in [r], j\in [m]$.
Given a deposition sequence $D$, we say that a subsequence $D'$ of $D$ is \emph{primal} if it is obtained from $D$ by deleting all characters which are associated with a trivial mask.
Notice that the border length of each mask associated with each character in a primal sequence is at least one, and the border length of all trivial masks is $0$. For the purpose of providing concise running times, we use $n$ to denote the size of the input. 

\begin{observation}
\label{obs:num_prim_dep_seq}
For any instance of $\PBMP$ and $\BMP$, the number of primal sequences is bounded by $\sum_{i=1}^{o}c^i\leq o\cdot c^o$.
\end{observation}

Additionally, since the number of ``borders'' between distinct probes is bounded from below by the number of distinct probes, we obtain:

\begin{observation}
\label{obs:bounded_num_probes}
Given a multiset $\probeset$ of probes. 
For any \yes-instance of $\PBMP$ and $\BMP$ over $\probeset$, the number of distinct probes in $\probeset$ is upper-bounded by $o+1$.
\end{observation} 

\begin{lemma}
\label{lem:primDepSeq2DepSeq}
For any instance of $\PBMP$ and $\BMP$, any primal sequence $D'$ corresponds to at most one good deposition sequence $D$. Furthermore, there exists an algorithm which runs in time $\bigO{o\cdot n}$ and which either computes this $D$ from $D'$ or correctly outputs that no such $D$ exists.
\end{lemma}

\begin{proof}

We provide the polynomial time algorithm to compute $D$ from $D'$; uniqueness follows by the fact that the algorithm is deterministic.

\begin{tabbing}
\textsc{Algorithm}$(D')$\\
1 \quad \= $(i:=1)$\\
2 \> \= Check whether a trivial mask for any character $x\in \Sigma$ can be applied. \\
3 \> \quad \= If not, go to 5.\\
4 \>  \> \= If yes, apply it, set $D:=D+x$, and go to 2.\\
5 \> \= Apply the mask for $D'[i]$. Set $D:=D+D'[i]$.\\
6 \> \= $i:=i+1$.\\
7 \> \= If $(i\leq |D'|)$ then go to 2.\\
8 \> \= Check whether a trivial mask for any character $x\in \Sigma$ can be applied. \\
9 \> \quad \= If not, go to 11.\\
10 \> \> \= If yes, apply it, set $D:=D+x$, and go to 8.\\
11 \> \= If there remains a nonempty probe $s$, then \textbf{reject}.\\
12 \> \= \textbf{Output} $D$.\\
\end{tabbing}

The algorithm runs in time $O(|D'|\cdot (c+|\probeset|\cdot max_{s\in \probeset}|s|))=O(o\cdot n)$. Correctness follows from the definition of primal sequences.
\end{proof}

\begin{theorem}
\label{thm:PBMPco}
$\PBMP_{c,o}$ is fixed-parameter tractable, and there exists an algorithm for $\PBMP_{c,o}$ which runs in time $\bigO{oc^o\cdot (n+o^2)}$.
\end{theorem}
\begin{proof}
This algorithm builds upon Observation~\ref{obs:num_prim_dep_seq}.
We can branch on all primal sequences.
For each candidate sequence $D'$ we check whether the primal sequence corresponds to a deposition sequence $D$ via Lemma~\ref{lem:primDepSeq2DepSeq}. For each such $D$, we compute and store $\borderlen(\placement,D)$. 
Finally, a solution with a minimum $\borderlen(\placement,D)$ is selected.
Observe that an applicable trivial mask can be found in linear time.
Along with Observation~\ref{obs:bounded_num_probes}, this yields a total runtime of $\bigO{oc^o\cdot (n+o^2)}$ by Proposition~\ref{prop:linear} and Lemma~\ref{lem:primDepSeq2DepSeq}.
\end{proof}

\subsection{Fpt-Algorithm for $\BMP_{c,o}$}
\label{sub:fptbmpco}

For a multiset $\probeset$ and $s\in \probeset$, we denote by $\probeset^{-s}$ the set of sequences in $\probeset$ which are distinct from $s$. An instance $(\probeset, r, m,o)$ of $\BMP_{c,o}$ is then called \emph{$s$-enveloped} if $|\probeset^{-s}|\leq o^2$.

\begin{lemma}
\label{lem:sea}
Any instance $(\probeset, r, m,o)$ of $\BMP_{c,o}$ such that $r>o$ and $m>o$ which is not $s$-enveloped for any $s\in \probeset$ is a no-instance.
\end{lemma}

\begin{proof}
Consider any placement $\placement$. For $s\in \probeset$, we say that a column (or row) is \emph{$s$-uniform} (w.r.t. $\placement$) if all cells in the column (or row) are only assigned sequences which are not distinct from $s$. Furthermore, we say that a column (or row) is \emph{uniform} if all cells in the column (or row) are not distinct from some sequence in $\probeset$.

Each non-uniform column and each non-uniform row contains at least one tuple of neighboring distinct sequences, which (regardless of $D$) contributes to an increase of $\borderlen(\placement,D)$ by at least $1$. Hence any solution $(\placement,D)$ of $(\probeset, r, m,o)$ must contain at most $o$ rows and at most $o$ columns which are not uniform. Furthermore,  
if there exists an $s$-uniform column (or row) for some $s\in \probeset$, then all other uniform columns (rows) must also be $s$-uniform---otherwise $\placement$ would contain more than $o$ non-uniform rows (columns), which we have already argued cannot happen.

To complete the proof, consider the possible cells where a sequence which is distinct from $s$ may appear. Clearly such sequences may only appear in the at most $o$ non-uniform columns and in the at most $o$ non-uniform rows, and these intersect in at most $o^2$ cells.
\end{proof}

We now consider two specific subcases of the problem before giving the theorem.

\begin{lemma}
\label{lem:BMPcoone}
There is an algorithm which solves any instance $(\probeset, r, m,o)$ of $\BMP_{c,o}$ such that $m>2o$ and $r>2o$ in time $\bigO{o^3\cdot c^o\cdot (n+o^2)}$.
\end{lemma}

\begin{proof}
By Lemma \ref{lem:sea}, there is either a sequence $s\in \probeset$ which represents the majority of sequences in $\probeset$, or $(\probeset, r, m,o)$ is a no-instance; since only at most one quarter of sequences in $\probeset$ are distinct from $s$, the sequence $s$ is unique and can be computed in time $|\probeset|$. 

Next, by Corollary \ref{cor:consecutive} (and the symmetric statement for rows), we can assume without loss of generality that all $s$-uniform columns and all $s$-uniform rows are placed consecutively in $\placement$. Notice that in this case only the first and last $o$ columns and rows can be non-$s$-uniform. Since any sequence $q$ distinct from $s$ can only be placed in columns and rows that are not $s$-uniform, the number of possibilities for $\placement(q)$ is bounded by $4o^2$. 

We now summarize the algorithm. First, we find $s$ in time $|\probeset|$. Second, for each of the at most $o^2$ sequences $q$ distinct from $s$ we branch on the at most $4o^2$ possible values of $\placement(q)$, resulting in a placement $\placement$. Third, for each such choice of $\placement$ we use the algorithm for $\PBMP_{c,o}$ from Theorem~\ref{thm:PBMPco} to find an optimal deposition sequence $D$ and store the obtained $\borderlen(\placement,D)$. Finally, we choose a tuple $(\placement,D)$ with a minimum $\borderlen(\placement,D)$. The bound on the running time follows from Theorem~\ref{thm:PBMPco}.
\end{proof}

\begin{lemma}
\label{lem:BMPcotwo}
There is an algorithm which solves any instance $(\probeset, r, m,o)$ of $\BMP_{c,o}$ such that $m>2o$ and $r\leq 2o$ in time $n\cdot c^{o^{\bigO{o}}}$.
\end{lemma}

\begin{proof}
By Observation~\ref{obs:bounded_num_probes}, we obtain that the number of distinct column placements is bounded by $o^r\leq o^{2o}$.

Now we reuse the algorithm given in the proof of Theorem~\ref{thm:BMPclr} with the only difference that in Step~\ref{step1} we branch on primal sequences and compute the corresponding (good) deposition sequence in polynomial time. The number of primal sequences is bounded by $o\cdot c^o$ (Observation~\ref{obs:num_prim_dep_seq}), the time required to compute the corresponding deposition sequence is bounded $O(o\cdot n)$ by Lemma~\ref{lem:primDepSeq2DepSeq}. For each fixed deposition sequence, the running time of steps 2--4 of the algorithm in Theorem \ref{thm:BMPclr} is bounded by $c^{o^{O(o)}}$, and hence the runtime bound of $o^{2o}\cdot (o\cdot n+n\cdot c^{o^{O(o)}})=n\cdot c^{o^{O(o)}}$.
\end{proof}

\begin{theorem}
\label{thm:BMPco}
$\BMP_{c,o}$ is fixed parameter tractable, and there exists an algorithm for $\BMP_{c,o}$ which runs in time $n\cdot c^{o^{\bigO{o}}}$.
\end{theorem}
\begin{proof}
In case $m>2o$ and $r>2o$ we use the algorithm described in the proof of Lemma~\ref{lem:BMPcoone}.
In case $m>2o$ and $r\leq 2o$ (or, by symmetry, if $m\leq 2o$ and $r>2o$) we use the algorithm described in the proof of Lemma~\ref{lem:BMPcotwo}.
In case $m\leq 2o$ and $r\leq 2o$ we branch over all of the at most $(4o^2)!$ placements $\placement$, resulting in at most $(4o^2)!$ instances of $\PBMP_{c,o}$ which can be solved individually in time $\bigO{oc^o\cdot (n+o^2)}$ by Theorem~\ref{thm:PBMPco}.
\end{proof}

\section{Conclusion}
\label{sec:conclusions}
In this work we considered the parameterized complexity of $\BMP$ and $\PBMP$, two fundamental problems related to the optimal design of microarrays, with respect to combinations of parameters centered around the number of distinct characters $c$.
We presented fpt-algorithms for both $\BMP$ and $\PBMP$ if the maximum probe length and the number of rows are viewed as additional parameters ($c,\ell,r$); and if the border length is the additional parameter ($c,o$).
In addition, we showed that $\PBMP$ parameterized by $c$ and $\ell$ is in FPT.
For $c,r$ (and also $c$ alone) we showed \paraNP-hardness for both $\BMP$ and $\PBMP$.
Hence, under the usual complexity theoretic assumptions, one cannot hope to find an fpt-algorithm for these~settings.

On our agenda for future work is to settle the question whether there is an fpt-algorithm for $\BMP$, parameterized by $c,\ell$. Another direction for future research is to study further (structural) parameters for these two problems. Furthermore, in our complexity analysis we plan to consider more sophisticated target functions that take other criteria in addition to the border length into account.

\bibliographystyle{abbrv}
\bibliography{literature-arxiv} 

\end{document}